\documentclass{IEEEtran}
\usepackage{graphicx}
\usepackage{mathptmx}
\usepackage{amsmath}
\usepackage{amssymb}
\usepackage{amsthm}
\usepackage{lipsum,color}

\usepackage{booktabs,caption}
\usepackage[flushleft]{threeparttable}
\theoremstyle{plain}
\newtheorem{theorem}{Theorem}[section]
\newtheorem{lemma}[theorem]{Lemma}

\newtheorem{corollary}{Corollary}[section]

\theoremstyle{definition}

\newtheorem{example}{Example}[section]
\newtheorem{problem}{Problem}[section]

\usepackage{cite}

\usepackage{textcomp}
\def\BibTeX{{\rm B\kern-.05em{\sc i\kern-.025em b}\kern-.08em
    T\kern-.1667em\lower.7ex\hbox{E}\kern-.125emX}}
\begin{document}
\title{{Galois Hull Dimensions of Gabidulin Codes}
\author{Habibul Islam and Anna-Lena Horlemann}

\thanks{H. Islam and A-L. Horlemann are with School of Computer Science, University of St. Gallen, Switzerland,  (e-mail: habibul.islam@unisg.ch, anna-lena.horlemann@unisg.ch). }}

\maketitle

\begin{abstract}
For a prime power $q$, an integer $m$ and $0\leq e\leq m-1$ we study the $e$-Galois hull dimension of Gabidulin codes $G_k(\boldsymbol{\alpha})$ of length $m$ and dimension $k$ over $\mathbb{F}_{q^m}$. Using a self-dual basis $\boldsymbol{\alpha}$ of $\mathbb{F}_{q^m}$ over $\mathbb{F}_q$, we first explicitly compute the hull dimension of $G_k(\boldsymbol{\alpha})$. Then a necessary and sufficient condition of $G_k(\boldsymbol{\alpha})$ to be linear complementary dual (LCD), self-orthogonal and self-dual will be provided. We prove the existence of $e$-Galois (where $e=\frac{m}{2}$) self-dual Gabidulin codes of length $m$ for even $q$, which is in contrast to the known fact that Euclidean self-dual Gabidulin codes do not exist for even $q$. 
As an application, we construct two classes of entangled-assisted quantum error-correcting codes (EAQECCs) whose parameters have more flexibility compared to known codes in this context.
\end{abstract}


\section{Introduction}
For a linear code $C$, the hull is defined by $\mathrm{hull}(C)=C\cap C^{\perp}$, where $C^{\perp}$ is the dual code of $C$. The hull investigation is important for many reasons, for instance, codes having $\mathrm{hull}(C)=\{0\}$ are said to be linear complementary dual (LCD) codes, introduced by Massey \cite{Massey92}, and can be applied for preventing popular cryptographic attacks such as side channel and fault-injection attacks \cite{Carlet16}. Codes satisfying $\mathrm{hull}(C)=C$ (resp. $C^{\perp}$) are known as self-orthogonal (resp. dual-containing), in particular, codes having $\mathrm{hull}(C)=C=C^{\perp}$ are called self-dual. These codes have some interesting applications in secret-sharing schemes and quantum coding, see \cite{Dougherty08,Jin10,Jin17}. Besides, the notion of hull dimension can be successfully applied for constructing quantum error-correcting codes \cite{Gao21,Liu19,Luo19,Sok22a}, and determining the computational complexity of algorithms for finding automorphism groups \cite{Leon82}, or checking the equivalence (permutation) of two linear codes \cite{Sendrier00,Leon91}. Several studies on the hull dimension under popular inner products (e.g., Euclidean, Hermitian, etc.) can be found in the literature. The $e$-Galois inner product (which generalizes both the Euclidean and the Hermitian) over $\mathbb{F}_{q^m}$ where $0\leq e\leq m-1$ has been introduced by Fan and Zhang in 2017 \cite{Fan17}. This inner product can also be applied for constructing quantum codes. However, so far there are only few studies on the $e$-Galois hull dimension (of generalized Reed-Solomon codes) and their application in quantum codes \cite{Cao22,Cao20,Fang22}. 
 
A special class of maximum distance separable (MDS) codes, called Gabidulin codes was introduced in \cite{Gabidulin85}. Detailed surveys of Gabidulin codes and their (Euclidean/Hermitian) duality can be found in \cite{Cruz18,Cruz21,Gorla21,Horlemann21,Ravagnani}.
Here, we investigate the $e$-Galois hull dimension of Gabidulin codes, in particular, by using a self-dual basis of $\mathbb{F}_{q^m}$ over $\mathbb{F}_q$, we obtain a necessary and sufficient condition for these codes to be LCD, self-orthogonal and self-dual. Unlike Euclidean self-dual Gabidulin codes \cite{Nebe16}, for $e=m/2$ we show that $e$-Galois self-dual Gabidulin codes exist for even $q$. Furthermore, we can construct MDS codes in $\mathbb{F}_{q^m}^m$ (whose $e$-Galois hull dimensions are calculated explicitly) of dimension $k>\lfloor \frac{p^e+m}{p^e+1}\rfloor$, where $p$ is the characteristic of the field, which was proposed as an open problem in \cite{Cao22,Cao20,Fang22}. At the end, using the $e$-Galois hull dimensions we obtain 
some classes of entangled-assisted quantum error-correcting codes (EAQECCs), among which some have dimensions $k>\lfloor \frac{p^e+m}{p^e+1}\rfloor$, and are hence not included in the codes that appeared in \cite{Cao22,Cao20,Fang22}.

\section{Preliminaries}	
Let $q$ be a prime power and $\mathbb{F}_{q^m}$ be a field of order $q^m$. Let the elements of $\boldsymbol{\alpha}=(\alpha_1,\cdots, \alpha_m)$ be a basis of $\mathbb{F}_{q^m}$ over $\mathbb{F}_q$, and denote by $\boldsymbol{\alpha}^q :=(\alpha_1^q,\cdots, \alpha_m^q)$ the coordinate-wise application of the Frobenius map. The \emph{Gabidulin code}  $G_k(\boldsymbol{\alpha})$ of length $m$ and dimension $k$ is defined by the generator matrix
\begin{equation}\label{eq1}
    G=\begin{pmatrix}
    \alpha\\
    \alpha^{q}\\
        \alpha^{q^2}\\
\vdots\\
    \alpha^{q^{k-1}}
\end{pmatrix},
\end{equation}
see \cite{Gabidulin85}. The Gabidulin code $G_k(\boldsymbol{\alpha})$ can also be defined as an evaluation code of linearized polynomials with evaluation points $\{\alpha_1,\cdots, \alpha_m\}$. It is a maximum rank distance (MRD) code, and hence also an MDS code, i.e., it has minimum Hamming distance $d=m-k+1$. For more details on Gabidulin codes see \cite{Gabidulin85,Lunardona18,Sheekey16}.\\
Let $0\leq e\leq m-1$ and $\boldsymbol{x}=(x_0, x_1, \dots, x_{m-1})$, $\boldsymbol{y}=(y_0, y_1, \dots, y_{m-1})$ $\in \mathbb{F}_{q^m}^m$, we recall from \cite{Fan17} that the \textit{$e$-Galois inner product} is defined by $$\boldsymbol{x}\cdot_e \boldsymbol{y}: = \sum_{i=0} ^{m-1} x_i y_i^{q^e}.$$ For $e=0$, it is the Euclidean and for $e=\frac{m}{2}$ (when $m$ is even), it is equivalent to the Hermitian inner product. \\
For a Gabidulin code $G_k(\boldsymbol{\alpha})$ of length $m$ over $\mathbb{F}_{q^m}$, the \emph{$e$-Galois dual} is defined by
\begin{align*}
    G_k(\boldsymbol{\alpha})^{\perp_e} :=\{\boldsymbol{x}\in \mathbb{F}_{q^m}^m : \boldsymbol{x}\cdot_e \boldsymbol{y}=0 \text{ for all }\boldsymbol{y}\in G_k(\boldsymbol{\alpha})\}.
\end{align*}
Further, $G_k(\boldsymbol{\alpha})$ is said to be self-orthogonal (resp. self-dual) if $G_k(\boldsymbol{\alpha})\subseteq G_k(\boldsymbol{\alpha})^{\perp_e}$ (resp. $G_k(\boldsymbol{\alpha})= G_k(\boldsymbol{\alpha})^{\perp_e}$).

Let $\boldsymbol{\alpha}=\{\alpha_1,\cdots, \alpha_m\}$ and $\boldsymbol{\beta}=\{\beta_1, \cdots,\beta_m\}$ be two bases of $\mathbb{F}_{q^m}$ over $\mathbb{F}_q$. Then $\beta$ is said to be a \emph{dual (orthogonal) basis} of $\boldsymbol{\alpha}$ if $Tr(\alpha_i\beta_j)=\delta_{ij}$, where $Tr: \mathbb{F}_{q^m}\longrightarrow \mathbb{F}_q$ is the usual trace function. When $\boldsymbol{\alpha}=\boldsymbol{\beta}$, we call $\boldsymbol{\alpha}$ a \emph{self-dual basis}.

Self-dual bases do not exists generally, the necessary and sufficient condition for their existence is given by the next result.
\begin{lemma} \cite{Zhe-Xian Wan}\label{lemma exist}
The vector space $\mathbb{F}_{q^m}$ over $\mathbb{F}_q$ has a self-dual basis if and only if $q$ is even or $q$ and $m$ both are odd.
\end{lemma}
We recall from \cite[page 22-24]{Gabidulin21} that for a Gabidulin code $G_k(\boldsymbol{\alpha})$ of length $m$ whose generator matrix is given by \eqref{eq1}, there is a unique basis $\boldsymbol{\beta}=(\beta_1,\beta_2,\cdots, \beta_m)$ such that the transpose of the Moore matrix $$B=\begin{pmatrix}
    \boldsymbol{\beta}\\
    \boldsymbol{\beta}^{q}\\
\vdots\\
    \boldsymbol{\beta}^{q^{m-1}}
\end{pmatrix}$$ is the inverse to the matrix $G$. Therefore, $GB^\top=B^\top G=I$. Note that $B^\top G=I$ is the matrix representation of the original definition of dual basis given above, i.e., $\beta$ is the dual basis to $\alpha$. This dual basis defines the parity check matrix (i.e., the generator matrix of the Euclidean dual code) of $G_k(\boldsymbol{\alpha})$ as 
\begin{equation}\label{eq2}
    H=\begin{pmatrix}
    \boldsymbol{\beta}^{q^{k}}\\
    \boldsymbol{\beta}^{q^{k+1}}\\
\vdots\\
    \boldsymbol{\beta}^{q^{m-1}}
\end{pmatrix}.
\end{equation}
Using the matrix $H$ we now compute the generator matrix for $G_k(\boldsymbol{\alpha})^{\perp_e}$ in the following.

\begin{theorem}
Let $G_k(\boldsymbol{\alpha})$ be a Gabidulin code of length $m$ and dimension $k$. Let $\boldsymbol{\beta}$ be the dual basis to $\boldsymbol{\alpha}$. Then the generator matrix of $G_k(\boldsymbol{\alpha})^{\perp_e}$ is
$$G^{\perp_e}=\begin{pmatrix}
    \beta^{q^{k-e}}\\
    \beta^{q^{k-e+1}}\\
\vdots\\
    \beta^{q^{m-e-1}}
\end{pmatrix}
=H^{q^{m-e}}
,$$
where $H$ is as in \eqref{eq2}.
\end{theorem}
\begin{proof}
We have 
$$G \cdot_e \big(G^{\perp_e}\big)^\top = G \cdot ((H^{q^{m-e}})^{q^e})^\top = G \cdot H^\top = 0.$$
Since $G^{\perp_e}$ is the generator matrix of the Gabidulin code $G_{m-k}(\boldsymbol{\beta}^{q^{k-e}})$, it has full rank and the statement follows.
\end{proof}

\section{Hull dimension of $G_k(\boldsymbol{\alpha})$}
Recall that the $e$-Galois hull of $G_k(\boldsymbol{\alpha})$ is defined by $\mathrm{Hull}_e\big(G_k(\boldsymbol{\alpha})\big)=G_k(\boldsymbol{\alpha})\cap G_k(\boldsymbol{\alpha})^{\perp_e}.$ Clearly, $\mathrm{Hull}_e\big(G_k(\boldsymbol{\alpha})\big)$ is a subspace of both $G_k(\boldsymbol{\alpha})$ and $G_k(\boldsymbol{\alpha})^{\perp_e}$. Therefore, $G_k(\boldsymbol{\alpha})$ is LCD if and only if $\dim \Big( \mathrm{Hull}_e\big(G_k(\boldsymbol{\alpha})\big)\Big)=0$, and self-orthogonal if and only if $\dim \Big( \mathrm{Hull}_e\big(G_k(\boldsymbol{\alpha})\big)\Big)=\dim \Big(G_k(\boldsymbol{\alpha})\Big)=k$. In the following we compute the hull dimension of Gabidulin codes.

\begin{theorem}\label{theorem hull1}
Let $\alpha$ be a self-dual basis of $\mathbb{F}_{q^m}$ over $\mathbb{F}_q$, and $G_k(\alpha)$ be a Gabidulin code of length $m$ and dimension $k$. Then
\begin{equation*}
        \dim \Big( \mathrm{Hull}_e\big(G_k(\boldsymbol{\alpha})\big)\Big)=\begin{cases}
      \min(m-k,e)  & \text{if $0\leq e\leq k$}\\
      \min(m-e,k) & \text{if $k+1\leq e\leq m-1$}
    \end{cases}.  
    \end{equation*}
\end{theorem}
\begin{proof}
Let $0\leq e\leq k$. Thus $q^{k-e}\geq {q^0}$ and hence
\begin{align*}
    \mathrm{rk} \begin{pmatrix} G \\\hline G^{\perp_e} \end{pmatrix} =\mathrm{rk} \begin{pmatrix}
    \alpha\\
    \alpha^{q}\\
    \vdots\\
    \alpha^{q^{k-1}}\\\hline\\
        \alpha^{q^{k-e}}\\
    \alpha^{q^{k-e+1}}\\
\vdots\\
        \alpha^{q^{m-e-1}}
\end{pmatrix} 
&=
    \mathrm{rk}\begin{pmatrix}
    \alpha\\
    \alpha^{q}\\
\vdots\\
        \alpha^{q^{\max(k-1,m-e-1)}}
\end{pmatrix}\\
&
=\max(k,m-e).
\end{align*}
Therefore, 
$$\dim \Big( \mathrm{Hull}_e\big(G_k(\boldsymbol{\alpha})\big)\Big)=m-\max(k,m-e)=\min(m-k,e).$$

On the other hand, let $e\geq k+1 $. So $q^{k-e}< {q^0}$ and hence 
\begin{align*}
    \mathrm{rk} \begin{pmatrix} G \\\hline G^{\perp_e} \end{pmatrix} 
=
    \mathrm{rk}\begin{pmatrix}
    \alpha^{q^{m-(e-k)}}\\
    \alpha^{q^{m-(e-k)+1}}\\
\vdots\\
\alpha\\
\alpha^q\\
\vdots\\
        \alpha^{q^{\max(k-1,m-e-1)}}
\end{pmatrix}
=\max(e,m-k).
\end{align*}
Thus, 
$$\dim \Big( \mathrm{Hull}_e\big(G_k(\boldsymbol{\alpha})\big)\Big)=m-\max(e,m-k)=\min(m-e,k).$$
\end{proof}
\begin{corollary}\label{cor lcd}
Let $\alpha$ be a self-dual basis of $\mathbb{F}_{q^m}$ over $\mathbb{F}_q$, and $G_k(\boldsymbol{\alpha})$ be a Gabidulin code of length $m$ and dimension $1\leq k\leq m-1$. Then $G_k(\boldsymbol{\alpha})$ is LCD if and only if $e=0$.
\end{corollary}
\begin{proof}
Let $G_k(\boldsymbol{\alpha})$ be a Gabidulin code of length $m$ and dimension $1\leq k\leq m-1$.
Let $e=0$. Then by Theorem \ref{theorem hull1} we get $\dim \Big(\mathrm{Hull}_e\big(G_k(\boldsymbol{\alpha})\big)\Big)=\min (m-k,e)=0$.
Hence $G_k(\boldsymbol{\alpha})$ is LCD. 

Conversely, let $G_k(\boldsymbol{\alpha})$ be LCD. Then $\dim \Big(\mathrm{Hull}_e\big(G_k(\boldsymbol{\alpha})\big)\Big)=0$. Thus, $k+1\leq e\leq m-1$ will never be the case. On the other hand, when $0\leq e\leq k$, $\min (m-k,e)=0$ implies $e=0$, since $m-k>0$.
\end{proof}

\begin{corollary}\label{cor self-orth0-one}
Let $\boldsymbol{\alpha}$ be a self-dual basis of $\mathbb{F}_{q^m}$ over $\mathbb{F}_q$, and $G_k(\boldsymbol{\alpha})$ a Gabidulin code of length $m$ and dimension $k\leq m-k$. Then $G_k(\boldsymbol{\alpha})$ is $e$-Galois self-orthogonal if and only if either $e=k$ or $k+1\leq e\leq leq m-k$.
\end{corollary}
\begin{proof}
 Let $G_k(\boldsymbol{\alpha})$ be $e$-Galois self-orthogonal. Then $G_k(\boldsymbol{\alpha})\subseteq G_k(\boldsymbol{\alpha})^{\perp_e}$ and hence $\dim \Big(\mathrm{Hull}_e\big(G_k(\boldsymbol{\alpha})\big)\Big)=k$. We distinguish two cases:
 
 Case (a): When $0\leq e\leq k$, we have by Theorem \ref{theorem hull1} that $\dim \Big(\mathrm{Hull}_e\big(G_k(\boldsymbol{\alpha})\big)\Big)=\min (m-k,e)=e=k$, since $e\leq k\leq m-k$. 
 
 Case (b): When $k+1\leq e\leq m-1$, we have by Theorem \ref{theorem hull1} that $\dim \Big(\mathrm{Hull}_e\big(G_k(\boldsymbol{\alpha})\big)\Big)=\min (m-e,k)=k$ which implies $m\geq k+e$. 
 
Conversely, for (a), let $e=k\leq \frac{m}{2}$. Now, Theorem \ref{theorem hull1} gives $\dim \Big(\mathrm{Hull}_e\big(G_k(\boldsymbol{\alpha})\big)\Big)=\min (m-k,e)=e=k$. Therefore, $G_k(\boldsymbol{\alpha})$ is $e$-Galois self-orthogonal.

In case (b), let $k+1\leq e\leq m-k$. Then Theorem \ref{theorem hull1} gives $\dim \Big(\mathrm{Hull}_e\big(G_k(\boldsymbol{\alpha})\big)\Big)=\min (m-e,k)=k$. Thus, $G_k(\boldsymbol{\alpha})$ is $e$-Galois self-orthogonal.
\end{proof}

\begin{corollary}\label{cor herr}
Let $m$ be even. Let $\alpha$ be a self-dual basis of $\mathbb{F}_{q^m}$ over $\mathbb{F}_q$, and $G_k(\alpha)$ be a Gabidulin code of length $m$, dimension $k=\frac{m}{2}$. Then $G_k(\alpha)$ is $e$-Galois self-dual if and only if $e=\frac{m}{2}$.
\end{corollary}
\begin{proof}
It follows from Corollary \ref{cor self-orth0-one} with $e=k=\frac{m}{2}$.
\end{proof}
Before we give an example we need to recall a known fact.
\begin{lemma}\cite{Carlet18,Massey92}\label{lemm known}
Let $C$ be a linear code of length $m$ over $\mathbb{F}_q$ with the generator matrix $G$. 
\begin{enumerate}
    \item $C$ is Euclidean LCD if and only if $GG^\top$ is non-singular.
    \item $C$ is Hermitian self-dual if and only if $m$ is even and $G\overline{G}^\top$ is the zero matrix, where $\overline{G}:=G^{q^{m/2}}$.
\end{enumerate}
\end{lemma}

\begin{example}
Let $q=2, m=4$ and $\omega$ be a primitive element of $\mathbb{F}_{2^4}$, where $\omega^4=\omega+1$. Then $\{1,\omega,\omega^2,\omega^3\}=\{\beta_1,\beta_2,\beta_3,\beta_4\}$ is a basis of $\mathbb{F}_{2^4}$ over $\mathbb{F}_2$ . The matrix $M=(M_{ij})_{4\times 4}$ where $Tr(\beta_i\beta_j)=M_{ij}$ is given by 
$$M=\begin{pmatrix}
0&0&0&1\\
0&0&1&0\\
0&1&0&0\\
1&0&0&1
\end{pmatrix}.$$
Then we have $EME^\top=I_4$, where 
$$E=\begin{pmatrix}
1&0&0&1\\
0&1&0&1\\
0&0&1&1\\
0&1&1&1
\end{pmatrix}.$$
Therefore, the self-dual basis $\boldsymbol{\alpha}=\{\alpha_1,\alpha_2,\alpha_3,\alpha_4\}=\{1+\omega^3, \omega+\omega^3,\omega^2+\omega^3, \omega+\omega^2+\omega^3\}$ is calculated by $\alpha_i=\sum\limits_{j=1}^4 E_{ij}\beta_j$. The Gabidulin code $G_2(\boldsymbol{\alpha})$ based on the self-dual basis $\boldsymbol{\alpha}$ is defined by the generator matrix 
$$G=\begin{pmatrix}
1+\omega^3 & \omega+\omega^3 & \omega^2+\omega^3 &\omega+\omega^2+\omega^3 \\
1+\omega^6 & \omega^2+\omega^6 & \omega^4+\omega^6 &\omega^2+\omega^4+\omega^6
\end{pmatrix}=\begin{pmatrix}
\boldsymbol{\alpha} \\
\boldsymbol{\alpha}^2
\end{pmatrix}.$$ 
Now, we verify our results on the $e$-Galois hull for $e=0,1,2,3$.
\begin{enumerate}
    \item Let $e=0$. Then $GG^\top=\begin{pmatrix}
1&0\\
0&1
\end{pmatrix}$, a non-singular matrix. Thus, by Lemma \ref{lemm known} $G_2(\boldsymbol{\alpha})$ is an Euclidean LCD code, which verifies Corollary \ref{cor lcd}.
\item Let $e=1$. Then 
$G^{\perp_{1}}=\begin{pmatrix}
\boldsymbol{\alpha}^2 \\
\boldsymbol{\alpha}^{2^2}
\end{pmatrix}=\begin{pmatrix}
\boldsymbol{\alpha}^2 \\
\boldsymbol{\alpha}^{4}
\end{pmatrix}$, and hence 
$$\mathrm{rk} \begin{pmatrix} G \\\hline G^{\perp_1} \end{pmatrix}=\mathrm{rk}\begin{pmatrix}
\boldsymbol{\alpha} \\
\boldsymbol{\alpha}^2 \\
\hline
\boldsymbol{\alpha}^2 \\
\boldsymbol{\alpha}^{4}
\end{pmatrix}=3,$$ 
i.e., $\dim\Big(\mathrm{Hull}_1\big(G_k(\boldsymbol{\alpha})\big)\Big)=4-3=1=\min(m-k,e)$, which verifies Theorem \ref{theorem hull1}.

\item Let $e=2$. Then we have
\begin{align*}
    &\overline{G}=G^{2^2}=G^4\\
    &=\begin{pmatrix}
1+\omega^{12} & \omega^4+\omega^{12} & \omega^8+\omega^{12} &\omega^4+\omega^8+\omega^{12} \\
1+\omega^{24} & \omega^{8}+\omega^{24} & \omega^{16}+\omega^{24} &\omega^8+\omega^{16}+\omega^{24}
\end{pmatrix}
\end{align*}
and $G\overline{G}^\top=O_{2\times 2}.$ Therefore, by Lemma \ref{lemm known} $G_2(\boldsymbol{\alpha})$ is Hermitian self-dual, which verifies Corollary \ref{cor herr}.
\item Let $e=3$. Then $G^{\perp_{3}}=\begin{pmatrix}
\boldsymbol{\alpha}^{2^3} \\
\boldsymbol{\alpha}^{2^0}
\end{pmatrix}=\begin{pmatrix}
\boldsymbol{\alpha}^8 \\
\boldsymbol{\alpha}
\end{pmatrix}$
and hence 
$$\mathrm{rk} \begin{pmatrix} G \\\hline G^{\perp_3} \end{pmatrix}=\mathrm{rk}\begin{pmatrix}
\boldsymbol{\alpha} \\
\boldsymbol{\alpha}^2 \\
\hline
\boldsymbol{\alpha}^8 \\
\boldsymbol{\alpha}
\end{pmatrix}=3, $$ 
i.e., $\dim\Big(\mathrm{Hull}_3\big(G_k(\boldsymbol{\alpha})\big)\Big)=4-3=1=\min(m-e,k)$, which verifies Theorem \ref{theorem hull1}.
\end{enumerate}
\end{example}

To conclude this section we want to set our results into the context of known results and open problems in the literature. 
Recently, in \cite{Cao22,Cao20,Fang22}, the $e$-Galois hulls of generalized Reed-Solomon (GRS) codes were determined, whose dimensions are always upper bounded by $\lfloor \frac{p^e+n}{p^e+1}\rfloor$. As per our knowledge, it is still an open problem to find MDS codes of larger dimensions whose corresponding $e$-Galois hull dimensions are known.
\begin{problem}\cite{Cao22,Cao20,Fang22}\label{prob}
 Let $p$ be the characteristic of $\mathbb{F}_q$. 
 Can we construct MDS codes in $\mathbb{F}_q^n$ of dimension $k$, such that $\lfloor \frac{p^e+n}{p^e+1}\rfloor <k\leq \lfloor \frac{n}{2}\rfloor$, of which we can determine the dimensions of their $e$-Galois hulls?
\end{problem}
With our result we give a solution to Problem \ref{prob}, since Theorem \ref{theorem hull1} is valid for any $0\leq k\leq m$. To be precise we need $q$ and $m$ be such that a self-dual basis exists, but on the other hand our codes' dimensions are not upper bounded.
\vspace{0.2cm}
\\
\textbf{Solution to Problem \ref{prob}.} Let $q$ be even or both $q$ and $m$ be odd, and let $\alpha$ be a self-dual basis of $\mathbb F_{q^m}$ over $\mathbb{F}_q$. The Gabidulin code $G_k(\boldsymbol{\alpha})$ is an MDS code of length $m$ and dimension $k$ of which the dimensions of all $e$-Galois hulls are determined by Theorem \ref{theorem hull1}.

\vspace{0.1cm}

\begin{example}
Let $q^m=3^{15}$ and $m=15$. Following the existence result of self-dual bases given by Lemma \ref{lemma exist}, let $\boldsymbol{\alpha}$ be a self-dual basis of $\mathbb{F}_{3^{15}}$ over $\mathbb{F}_3$, and $G_k(\boldsymbol{\alpha})$ a Gabidulin code of length $15$ and arbitrary dimension $1\leq k\leq 15$. For $e\geq 3$, we get $\lfloor \frac{p^e+m}{p^e+1}\rfloor=1$. Hence, for $1<k\leq 15$, $G_k(\boldsymbol{\alpha})$ is an MDS code of dimension $k>\lfloor \frac{p^e+m}{p^e+1}\rfloor$ whose $e$-Galois hull dimensions are completely determined by Theorem \ref{theorem hull1}. Similarly, for $e=2$ (respectively $e=1$) we get previously unknown code parameters for $k>2$ (respectively $k>4$).
\end{example}

\section{Applications to EAQECCs}
As one possible application of Gabidulin codes (MDS by parameters), we now focus on constructing entangled-assisted quantum error-correcting codes (abbreviated as EAQECCs). Recall that an EAQECC $[[n,k,d;c]]_{q^m}$ refers to a $q^m$-ary quantum code (discovered by Shor \cite{Shor95} in 1995) of length $n$, dimension $k$ and minimum distance $d$ which pre-shares (between sender and receiver) $c$ maximally entangled states. In particular, for $c=0$, it is simply a standard quantum code. As it is known, a quantum code $[[n,k,d]]_{q^m}$ satisfies the Quantum Singleton Bound $2d\leq n-k+2$, and is known as a quantum MDS code when attains it with equality. Similarly for an EAQECC the Singleton Bound was established in \cite{Brun14} as 
\begin{align}\label{bound}
    2d\leq n-k+2+c.
\end{align}
Further, \cite{Lai18} validated the bound (\ref{bound}) when $d\leq \frac{n}{2}+1$, where $n$ is the code length. However, recently, Grassl \cite{Grassl21} computed a few examples which violate the above bounds for a certain range. Therefore, when EAQECCs attain the bound (\ref{bound}), it would be more appropriate calling them EAQECCs having relatively large minimum distance than MDS EAQECCs (although many researchers still prefer to use MDS EAQECCs). 

EAQECCs can be constructed from classical linear codes  based on the parity check matrix (see, \cite{Guenda20,Liu19}) and by using the $e$-Galois hull dimensions, see \cite{Cao22,Cao20,Fang22}.

\begin{lemma}\cite[Corollary IV.1]{Cao22}\label{lemma EAQECC}
For a classical MDS code $C$ of parameters $[n,k,d]_{q^m}$, there exists an EAQECC of parameters $[[n,k-\dim \big(\mathrm{Hull_e}(C)\big),d; n-k-\dim \big(\mathrm{Hull_e}(C)\big)]]_{q^m}$.
\end{lemma}

In the light of Lemma \ref{lemma EAQECC} and Theorem \ref{theorem hull1} we construct EAQECCs with relatively large minimum distance in the next results.
\begin{theorem}\label{theorem EAQ1}
Let $\alpha$ be a self-dual basis of $\mathbb{F}_{q^m}$ over $\mathbb{F}_q$, and $G_k(\alpha)$, a Gabidulin code of length $m$.
\begin{enumerate}
    \item For $0\leq e\leq k$ there exists an $[[m,k-\ell,m-k+1;m-k-\ell]]_{q^m}$ EAQECC, where $\ell=\min(m-k,e)$.
    \item For $k+1\leq e\leq m-1$ there exists an $[[m,k-\ell,m-k+1;m-k-\ell]]_{q^m}$ EAQECC, where $\ell=\min(m-e,k)$.
\end{enumerate}
\end{theorem}

It is worth mentioning that \cite{Cao22,Cao20,Fang22} also constructed  EAQECCs using the $e$-Galois hull dimension, but their codes have dimension $k$ upper bounded by $\lfloor \frac{p^e+m}{p^e+1}\rfloor$ or $\lfloor \frac{p^e+m-1}{p^e+1}\rfloor$, where $q=p^h$ and $2e\mid h$. Since in our case, dimension $k$ has no such restriction, the codes parameters in Theorem \ref{theorem EAQ1} can exceed this bound. Therefore, we have more flexible parameters for EAQECCs. Some explicit parameters for EAQECCs with relatively large minimum distances are listed in Table \ref{tab1}, where $k>\lfloor \frac{p^e+m}{p^e+1}\rfloor$.

\begin{table}
\renewcommand{\arraystretch}{1.5}
\begin{center}
\caption{Some EAQECCs $[[m,k,d;c]]_q$ satisfying $k>\lfloor \frac{p^e+m}{p^e+1}\rfloor$}
\vspace{0.5cm}

\begin{tabular}{|c|c|c|c|c|c|c|c|c|}

\hline
$(q=p^m,e)$ & EAQECC &  $k$ & $\lfloor \frac{p^e+m}{p^e+1}\rfloor$    \\
\hline

$(2^{100},2)$ & $[[100,k-2,101-k;98-k]]_{2^{100}}$ & $21\leq k\leq 98$ & $20$    \\

$(2^{100},2)$ & $[[100,2k-100,101-k;0]]_{2^{100}}$ & $ k\geq 98$ & $20$     \\

$(3^{67},40)$ & $[[67,0,68-k;67-2k]]_{3^{67}}$ & $2\leq k\leq 27$ & $1$    \\
$(3^{67},40)$ & $[[67,k-27,68-k;40-k]]_{3^{67}}$ & $27\leq k\leq 39$ & $1$     \\

\hline
\end{tabular}\label{tab1}
\end{center}
\end{table}

\section{Conclusion}
We investigated the $e$-Galois hulls of Gabidulin codes constructed from self-dual bases of $\mathbb{F}_{q^m}$ over $\mathbb{F}_q$. The dimensions of their $e$-Galois hull can be determined with a very easy formula, for any $0\leq e \leq m-1$. Since self-dual bases exist if and only if $q$ is even, or both $q$ and $m$ are odd, we have hence established a large family of MDS codes for which all $e$-Galois hull dimensions are known. This had previously been posed as an open problem for MDS codes of (relatively) large dimension. Furthermore, we can construct new entangled-assisted quantum error-correcting codes from these Gabidulin codes. 

It is interesting to see that the relation of a code to its $e$-Galois dual code depends heavily on the parameter $e$. In particular, we showed that Gabidulin codes generated by a self-dual (finite extension field) basis are always self-dual with respect to the Hermitian inner product, i.e., for $e=\frac{m}{2}$. Note that for the existence of a self-dual basis we need $q$ to be even in this case. In contrast, for the Euclidean inner product, i.e. $e=0$, it was shown in \cite{Nebe16} 
that self-dual Gabidulin codes (or any other maximum rank distance (MRD) codes) do not exist for even $q$. 

In this paper, we focused on Gabidulin codes $G_k(\boldsymbol{\alpha})$ generated by a self-dual basis vector $\boldsymbol{\alpha}$. In future work we will investigate the hull dimension of $G_k(\boldsymbol{\alpha})$ for other, non-self-dual, bases $\boldsymbol{\alpha}$ of $\mathbb{F}_{q^m}$ over $\mathbb{F}_q$. Moreover, we will study other properties of dual codes with respect to the $e$-Galois inner product, as e.g. MacWilliams identities.

\end{document}